\documentclass[12pt]{article}
\makeatletter
\@addtoreset{equation}{section}
\makeatother

\usepackage[pdftex]{graphicx}
\usepackage{subfigure}
\usepackage{array}
\usepackage{amsmath}

\usepackage{rotating}
\usepackage{longtable}
\usepackage{mathtools}
\usepackage{natbib}
\usepackage{makecell}

\usepackage{amsfonts}
\usepackage{array}
\usepackage{geometry}
\usepackage{color}
\usepackage{lipsum}

\usepackage{diagbox} 
\usepackage{pdfpages}
\usepackage{multicol}
\usepackage[utf8]{inputenc}
\usepackage[english]{babel}
\usepackage{mathtools}
\usepackage{bbm}
\usepackage{amssymb}
\usepackage{amsthm}
\usepackage{hyperref}
\usepackage{bm}
\usepackage{mathtools}
\usepackage{fancybox}
\usepackage{multirow}
\usepackage{authblk}
\usepackage{algorithm}
\usepackage{algpseudocode}
\usepackage{setspace}
\usepackage{caption}
\makeatletter
\newsavebox{\@brx}
\newcommand{\llangle}[1][]{\savebox{\@brx}{\(\m@th{#1\langle}\)}%
  \mathopen{\copy\@brx\kern-0.5\wd\@brx\usebox{\@brx}}}
\newcommand{\rrangle}[1][]{\savebox{\@brx}{\(\m@th{#1\rangle}\)}%
  \mathclose{\copy\@brx\kern-0.5\wd\@brx\usebox{\@brx}}}
\makeatother

\newcommand{\vertiii}[1]{{\vert\kern-0.25ex\vert\kern-0.25ex\vert #1\vert\kern-0.25ex\vert\kern-0.25ex\vert}}
\newcommand{\Vertiii}[1]{{\left\vert\kern-0.25ex\left\vert\kern-0.25ex\left\vert #1\right\vert\kern-0.25ex\right\vert\kern-0.25ex\right\vert}}
\newcommand{\tabincell}[2]{\begin{tabular}{@{}#1@{}}#2\end{tabular}}  

\usepackage{xcolor}
\renewcommand{\raggedright}{\leftskip=0pt \rightskip=0pt plus 0cm}
\title{Break Point Detection for Functional Covariance}
\author[1]{Shuhao Jiao\thanks{shjiaoqd@gmail.com}}
\author[2]{Ron D.\ Frostig\thanks{rfrostig@uci.edu}}
\author[1]{Hernando Ombao\thanks{hernando.ombao@kaust.edu.sa}}

\affil[1]{Statistics Program,  KAUST, Saudi Arabia}
\affil[2]{Department of Neurobiology and Behavior, UC Irvine, USA}

\date{}
\begin{document}
	\maketitle			
	\setlength\parindent{0pt}
	\setlength{\parskip}{1em}
	\theoremstyle{definition}
	\newtheorem{assumption}{Assumption}
	\newtheorem{theorem}{Theorem}
	\newtheorem{lemma}{Lemma}
	\newtheorem{prop}{ACRosition}
	\newtheorem{definition}{Definition}
	\newtheorem{corollary}{Corollary}
	\newtheorem{remark}{Remark}
\begin{abstract}
Many experiments record sequential trajectories where each trajectory consists of oscillations and fluctuations around zero. Such trajectories can be viewed as zero-mean functional data. When there are structural breaks (on the sequence of trajectories) in higher order moments, it is not always easy to spot these by mere visual inspection. {Motivated by this challenging problem in brain signal analysis}, we propose a detection and testing procedure to find the change point in functional covariance. 
The detection procedure is based on the cumulative sum statistics (CUSUM). The classical testing procedure for functional data depends on a null distribution which depends on {infinitely many} unknown parameters, though in practice only a finite number of these can be included for the hypothesis test of the existence of change point. This paper provides some theoretical insights on the influence of the number of parameters.
Meanwhile, the asymptotic properties of the estimated change point are developed. The effectiveness of the proposed method is numerically validated in simulation studies and an application to investigate changes in rat brain signals following an experimentally-induced stroke.\\

\end{abstract}

\section{Introduction}
Functional data analysis has attracted attention of researchers in the last few decades and many methods for structural break detection of functional data have been developed.  Here, we propose a method of detecting and testing structural breaks in covariance function. The motivation of this paper comes from a neuroscience experiment conducted in the Frostig neurobiology laboratory at 
UC Irvine (see Wann, 2017) to investigate the change in rat brain local field potentials following an induced ischemic stroke (by arterial clamping).  Local field potentials (LFPs) are recorded from 32 implanted micro-tetrodes during the pre-stroke and post-stroke phase (each phase consists of five minutes of recording) which are segmented into one-second epochs. Thus we have multivariate (32-dimensional) functional curves for each epoch and a total of 300 epochs for the pre-stroke phase and also 300 for the post-stroke phase. We expect to observe structural change in each LFP tetrode following the onset of simulated stroke. Here, all epoch trajectories fluctuate around $0$, leading to zero mean functions. In fact, it is common to preprocess and treat many brain electrical 
signals (e.g., electroencephalograms and local field potentials) to be random fluctuations around 0. {(see e.g., Ombao et al.~(2005, 2016), Motta and Ombao (2012), Fiecas and Ombao (2016), and Wu et al.~(2014))}. Thus our goal here is to develop a method for detecting the change point in the covariance function. One major benefit of developing the test procedure based on entire epochs, which are treated as random functions, is that the test procedure is robust to chance variation {(e.g., outliers or irregular extreme values)} and random errors because it is nearly impossible that all observations in an epoch are contaminated by chance variation. Moreover, the effect of chance variation or random errors can be attenuated by functional smoothing techniques.

There have been a number of methods developed for functional structural break analysis in mean function. In Berkes et al.~(2009), a testing procedure for change in the mean involves checking for structural break in the functional principal components; 
and Aue et al.~(2009a) quantified the large sample behavior of the change point estimator. Aston et al.~(2012a) extended the results to dependent functional data. Change-aligned principal components for such change point problems was developed in \cite{torgovitski2015detecting}, { to solve the problem that the leading principal component sometimes is not informative of the structural change.} Aue et al.~(2014) proposed a method to check the change point of coefficient operators in potentially non-homogeneous functional autoregressive model; and Aue et al.~(2018) proposed a fully functional detecting procedure without dimension reduction. 

There are also methods related to change point detection of covariance function or covariance matrix. Aue et al.~(2009b) studied the structural break detection problem for the covariance matrix of multivariate time series. They proposed to stack the lower triangular elements of covariance matrix and detect the structure break of the concatenated vectors.  Jaru\v skov\'a (2013) proposed a change point detection approach for functional covariance of i.i.d.~functions based on the truncated spectrum of functional covariance. Chen and Zhang (2015) proposed a novel graph-based change point detection framework, which can be applied to detect structural break in functional covariance if the graph is constructed on $\{Y_i(t)Y_i(s)\colon i\ge1\}$, where $\{Y_i(t)\colon i\ge1\}$ is the functional sequence with a change point in functional covariance.  {Avanesov and Buzun (2018) and Zhong et al.~(2019) studied the problem of change point detection of covariance matrix in a high-dimensional setting. Dette \& Kokot (2020) proposed a sup-norm approach. Aston et al.~(2012b),  Gromenko et al.~(2017) and Stoehr et al.~(2020) studied the structural break problem for bivariate or trivariate functions, specifically, spatial-temporal data and fMRI data, but they assume separability of the covariance function or apply separable fPCA, which are essentially based on low rank approximation of the covariance function, and could be overly restrictive for practical data analysis.} Aue et al.~(2020) dealt with analyzing structural break of spectrum and trace of covariance operator. {Harris et al.~(2021) proposed a scalable multiple change point detection procedure which also handles changes in variance. In contrast to these methods, we study the structural break for the complete covariance structure. No separable assumption is made thus making the proposed method suitable for a broad range of cases. }

There are other change point methods that can be applied to structural break detection for brain signals. Fryzlewicz \& Rao (2014) proposed the ``BASTA'' method for detecting multiple change points in the structure of an auto-regressive conditional heteroscedastic model. Kirch et al.~(2015) used VAR model to detect change points in multivariate time series and applied the method to EEG sequences. Cho \& Fryzlewicz (2015) proposed a sparsified binary segmentation method for the second-order structure of a multivariate time series. Schr{\"o}der \& Ombao (2017) proposed a FreSpeD method to detect the change point in the spectrum and coherence sequences of multivariate time series. Sundararajan \& Pourahmadi (2018) proposed a nonparametric method to detect multiple change points in multivariate time series based on difference in the spectral density matrices. {A general change point framework, Fr\'echet change point detection, was proposed in Dubey and M\"uller (2020).} 

Our proposed method can be used to detect the structural break in brain signals by checking the change point in the covariance function of epoch trajectories. In comparison to existing work, our method focuses on the ``big picture'' of brain signals, that is, we aim to find the change point in the sequences of functional epochs/trials instead of changes within an epoch. In addition to the robustness to chance variation, another advantage of our functional procedure lies in its ability to extract and use intra-curve information. As this new functional procedure checks the structural break of the entire covariance function, intra-curve information is incorporated, which can potentially reveal the structural break. This is discussed in more detail in the simulation studies. This paper provides a new perspective for change point problem of brain signal data. 

The major contribution of this article is developing a procedure to detect the change point in covariance function. We consider a general situation where functions are weakly dependent. Dimension reduction techniques, such as functional principal component analysis (fPCA), are very popular in functional data analysis. These techniques are able to extract the most important features, but may lead to loss of information. Indeed, this loss of information may not be crucial for functional reconstruction but could be critical for change point detection, especially when the leading principal components are orthogonal to the discrepancy. Note that even though the test procedure without dimension reduction avoids the loss of information, the null distribution still depends on {infinitely many} unknown parameters. Denoting $D_\rho$ to be the number of parameters included, Aue et al.~(2018) proposed to select $D_\rho$ to be the number of basis elements over which the initial discretely observed functional data are smoothed, but did not provide clear theoretical insights on this selection. Here, we provide some theoretical insights on the influence of the number of unknown parameters. These have not been previously discussed in the literature.

{The contribution of the work is summarized as follows:
\begin{itemize}
\item We study the CUSUM statistics for the change point detection problem of functional covariance of weakly dependent functional data, and establish complete estimation, detection and testing procedure, and the corresponding theoretical results.
\item The classical null distribution involves {infinitely many} unknown parameters, and is approximated by a truncated version. We study the convergence rate of the estimated truncated null distribution.
\item The work is motivated by the problem of detecting changes in brain signals. The brain signals at hand are LFPs which have zero mean and 
thus these signals are fluctuations around zero. The method provides complimentary information for structural break in brain signals.
\end{itemize}}

The rest of the article is organized as follows. In Section~\ref{s2}, we present some preliminaries of functional data. In Section \ref{s3}, we develop the change point model for covariance function, along with the procedure for estimation, detection and testing. We also derive the asymptotic properties of the proposed change point estimator. In Section~\ref{s4}, we report some simulation results. In Section~\ref{s5}, we analyze LFPs, and conclude in Section~\ref{s6}. Proofs of the theorems are in the supplement material.

\section{Preliminaries}
\label{s2}
{For a series of stationary random functions $\{Y_i(t)\colon t\in\mathcal{T}, i\in\mathbb{N}\}$} such that $\mathcal{T}\subset\mathbb{R}$ and $E\{\int_\mathcal{T} Y^2_i(t)dt\}<\infty$, 
{the mean function is defined as $\mu(t)=E\{Y_i(t)\},$ 
and the covariance operator and covariance function are defined respectively as
$$\Gamma(x)(t)=E\{\langle Y_i-\mu,x\rangle (Y_i(t)-\mu(t))\},\qquad C(t,s)=E\{(Y_i(t)-\mu(t))(Y_i(s)-\mu(s))\}.$$}
Define $X_i(t,s)=\{Y_i(t)-\mu(t)\}\{Y_i(s)-\mu(s)\},\ t,s\in\mathcal{T}$ as the data analogue of the covariance function, and denote $L^2(\mathcal{T}\times\mathcal{T})$ to be the space of square integrable functions defined over $\mathcal{T}\times\mathcal{T}$. {As the mean function $\mu(t)$ is unknown, $X_i(t,s)=\{Y_i(t)-\hat{\mu}(t)\}\{Y_i(s)-\hat{\mu}(s)\}$, where $\hat{\mu}(t)=N^{-1}\sum_{i=1}^{N}Y_i(t)$ is the sample average of $Y_1(t),\ldots,Y_N(t)$}. For any $X_i(t,s)$, $X_j(t,s)\in L^2(\mathcal{T}\times\mathcal{T})$, we define the inner product of the bivariate functions as 
$\llangle X_i,X_j\rrangle=\int_{\mathcal{T}\times\mathcal{T}}X_i(t,s)X_j(t,s)dtds,$
and the norm as 
{$\vertiii{X_i}^2=\int_{\mathcal{T}\times\mathcal{T}}X^2_i(t,s)dtds.$}

Obviously, $E\{X_i(t,s)\}=C(t,s)$.
In addition, we define the covariance and auto-covariance function of $X_i(t,s)$ as
$$\mathcal{\mathcal{C}}_{X,h}(t,s,\tilde{t},\tilde{s})=E\{(X_i(t,s)-C(t,s))(X_{i+h}(\tilde{t},\tilde{s})-C(\tilde{t},\tilde{s}))\},\qquad h\in\mathbb{N}.$$ 
The developed detection and test procedure involves the long-run covariance function of $\{X_i(t,s)\colon i\in\mathbb{N}\}$, defined as {the summation of all lagged covariance functions presented below}
$$LC_X(t,s,\tilde{t},\tilde{s})=\sum_{h=-\infty}^\infty \mathcal{\mathcal{C}}_{X,h}(t,s,\tilde{t},\tilde{s}),$$
and it is evident that $LC_X$ is a positive definite kernel in $L^2(\mathcal{T}\times\mathcal{T})$, and thus admits the following representation by Theorem 1.1 in Ferreira \& Menegatto (2009),
\begin{equation}
\label{eq1}
LC_X(t,s,\tilde{t},\tilde{s})=\sum_{d=1}^\infty\rho_d\psi_d(t,s)\psi_d(\tilde{t},\tilde{s}),
\end{equation}
where the bivariate eigenfunctions $\{\psi_d(t,s)\colon d\in\mathbb{N}_+\}$ form a series of orthonormal basis of $L^2(\mathcal{T}\times\mathcal{T})$, and the eigenvalues $\{\rho_d\colon d\in\mathbb{N}_+\}$ (in strictly descending order) account for the variation level of the principal components $\{\llangle X,\psi_d\rrangle\colon d\in\mathbb{N}_+\}$. 
\section{Main results}
\label{s3}
\subsection{Detection and testing procedure}
In the case of single change point, we assume the following change point model for the covariance function
\begin{equation*}
\label{eq2}
E\{X_{i}(t,s)\}=C^{(1)}(t,s)\mathbb{I}{(i\le k^*)}+ C^{(2)}(t,s)\mathbb{I}{(i>k^*)}.
\end{equation*}
We assume there is no structural break in the mean function. {This assumption is reasonable for many brain signals (e.g., local field potentials, EEG recordings), which always oscillate near zero}. The interest here is to test if the covariance function remains constant across $i$, specifically, we want to test the null hypothesis
$${H_0\colon C^{(1)}(t,s)=C^{(2)}(t,s)\ \mbox{for all}\ (t,s)}$$
against
$$H_a\colon C^{(1)}(t,s)\ne C^{(2)}(t,s)\ \mbox{for some}\ (t,s),$$
We assume that $\{Y_i(t)\colon i\in\mathbb{N}\}$ satisfy the following conditions.
\begin{assumption}
\label{a1}
There is a measurable function $f\colon S^\infty\to L^8(\mathcal{T})$, where $S$ is a measurable space, and i.i.d.~innovations $\{\epsilon_i\colon i\in\mathbb{N}\}$ taking values in $S$, so that under $H_0$, $Y_i(t)=f(\epsilon_i,\epsilon_{i-1},\ldots)$ and under $H_a$,
\begin{equation*}
Y_i(t)=\left\{
\begin{aligned}
f_1(\epsilon_i,\epsilon_{i-1},\ldots),\qquad i\le k^*\\
f_2(\epsilon_i,\epsilon_{i-1},\ldots),\qquad i>k^*
\end{aligned}
\right.
\end{equation*}
where $f_1,f_2$ are defined similarly with $f$. It is assumed that $E\|Y_i(t)\|^{8}_4<\infty$, where $\|\cdot\|_4$ denotes the $l^4$-norm. In addition, There exists a $m$-dependent sequence $\{Y_{i,m}(t)\colon i\in\mathbb{N}\}$, so that under $H_0$, $$Y_{i,m}(t)=f(\epsilon_i,\ldots,\epsilon_{i-m+1},\epsilon^*_{i-m},\epsilon^*_{i-m-1},\ldots),$$
and under $H_a$,
\begin{equation*}
Y_{i,m}(t)=\left\{
\begin{aligned}
f_1(\epsilon_i,\ldots,\epsilon_{i-m+1},\epsilon^*_{i-m},\epsilon^*_{i-m-1},\ldots),\qquad i\le k^*\\
f_2(\epsilon_i,\ldots,\epsilon_{i-m+1},\epsilon^*_{i-m},\epsilon^*_{i-m-1},\ldots),\qquad i>k^*
\end{aligned}
\right.
\end{equation*}
where $\epsilon_i^*$ is an independent copy of $\epsilon_i$, such that
$$\sum_{m=0}^\infty\{E\|Y_i(t)-Y_{i,m}(t)\|_4^{8}\}^{1/8}<\infty.$$
\end{assumption}
We now describe the CUSUM testing procedure. We first apply the detection procedure to find the change point candidate, and then apply the testing procedure to test the significance of the candidate. To proceed, we first introduce the estimators of the covariance function $C^{(1)}_k$ and $C^{(2)}_k$ for the segments $[1,k]$ and $[k+1,N]$ to be
$$\widehat{C}^{(1)}_k(t,s)=\frac{1}{k}\sum_{i=1}^kX_{i}(t,s),\qquad \widehat{C}^{(2)}_k(t,s)=\frac{1}{N-k}\sum_{i=k+1}^NX_{i}(t,s).$$
Under the null hypothesis, the difference $\widehat{C}^{(1)}_k(t,s)-\widehat{C}^{(2)}_k(t,s)$ should be close to zero for all $1< k<N$ and $(t,s)\in\mathcal{T}\times\mathcal{T}$. We incorporate a weight function to attenuate the end-point effect, and obtain the following weighted difference
$$\Delta_k(t,s)=\frac{k(N-k)}{N}\left\{\widehat{C}^{(1)}_k(t,s)-\widehat{C}^{(2)}_k(t,s)\right\}=\sum_{i=1}^kX_{i}(t,s)-\frac{k}{N}\sum_{i=1}^NX_{i}(t,s)$$
and large value of $\Delta_k(t,s)$ should be expected for some $k,t,s$ if structural break is present. The detection step is based on the following cumulative sum statistics (CUSUM)
$$T_N(\theta)=\frac{1}{N}\int\int\left\{\sum_{i=1}^{[N\theta]}X_{i}(t,s)-\frac{[N\theta]}{N}\sum_{i=1}^NX_{i}(t,s)\right\}^2dtds,$$
where $k=[N\theta]$.

To determine the change point candidate, we find the maximizer of $T_N(\theta)$. To ensure uniqueness, we define the change point candidate as 
$$\hat{\theta}^*_N=\inf\{\theta\colon T_N(\theta)=\sup_{0<\theta'<1}{T_N(\theta')}\}.$$
The next step is to apply a hypothesis test to classify the candidate change point as a change point or otherwise. The proposed test statistic is
$$T_N(\hat{\theta}^*_N)=\max_{0< \theta< 1}T_N(\theta).$$ The following theorems provide the asymptotic properties of the test statistics under $H_0$ and $H_a$.

\begin{theorem}
\label{th1}
Under Assumption \ref{a1} and $H_0$, 
$$T_N(\hat{\theta}^*_N)\overset{\mathcal{D}}\to\sup_{\theta\in[0,1]}\sum_{d=1}^\infty\rho_dB^2_{d}(\theta), \qquad N\to\infty.$$ 
where $\{B_d\colon d\in\mathbb{N}\}$ are i.i.d.~standard Brownian bridges defined on $[0,1]$.
\end{theorem}
\begin{remark}
{As a reminder, $\{\rho_d\colon d\ge1\}$ are the eigenvalues of $LC_X(t,s,\tilde{t},\tilde{s})$.}
\end{remark}
\begin{theorem}
\label{th2}
Under Assumption \ref{a1} and $H_a$, 
$T_N(\hat{\theta}^*_N)\to\infty,~\mbox{as}\ N\to\infty.$
\end{theorem}

{The null distribution incorporates {infinitely many} unknown eigenvalues $\rho_d$'s. In practice, the plug-in estimator $\sum\limits_{d=1}^{D_\rho}\hat{\rho}_dB^2_{d}(\theta)$ is employed instead. The existing literature does not provide theoretical insights on the influence of the selection of $D_\rho$. The selection of $D_\rho$ should trade off the balance between bias and variance. A large value of $D_\rho$ leads to small bias, but large estimation variance. We provide some theoretical insights on the selection of  $D_\rho$ in Section~\ref{s3.2}. The practical procedure of estimating $\{\rho_d\colon d\in\mathbb{N}\}$ is discussed in Section~\ref{s3.3}.}

\subsection{Selection of $\rho_d$'s}
\label{s3.2}
{One key step of the testing procedure is estimating the unknown eigenvalues of the long-run covariance function $LC_X(t,s,\tilde{t},\tilde{s})$. 
Under $H_0$, the (auto-)covariance of $\{X_i(t,s)\colon i\in\mathbb{N}\}$ is defined in Section~\ref{s2}. Under $H_a$, the (auto-)covariance is defined as 
$$\mathcal{C}_{X,h}(t,s,\tilde{t},\tilde{s})=\theta\mathcal{\mathcal{C}}^{(1)}_{X,h}(t,s,\tilde{t},\tilde{s})+(1-\theta)\mathcal{\mathcal{C}}^{(2)}_{X,h}(t,s,\tilde{t},\tilde{s}),$$
where, as $h\ge0$,
\begin{align*}
\mathcal{\mathcal{C}}^{(1)}_{X,h}(t,s,\tilde{t},\tilde{s})&=E\{(X_{i-h}(t,s)-C^{(1)}(t,s))(X_{i}(\tilde{t},\tilde{s})-C^{(1)}(\tilde{t},\tilde{s}))\},\qquad i\le k^*,\\
\mathcal{\mathcal{C}}^{(2)}_{X,h}(t,s,\tilde{t},\tilde{s})&=E\{(X_{i}(t,s)-C^{(2)}(t,s))(X_{i+h}(\tilde{t},\tilde{s})-C^{(2)}(\tilde{t},\tilde{s}))\},\qquad i> k^*,
\end{align*} 
and the case as $h<0$ can be defined similarly. Additionally, for the four-way function $\mathcal{C}(t,s,\tilde{t},\tilde{s})$, let
$${\vertiii{\mathcal{C}(t,s,\tilde{t},\tilde{s})}_\mathcal{S}^2=\int\int \mathcal{C}^2(t,s,\tilde{t},\tilde{s}) dtdsd\tilde{t}d\tilde{s}.}$$ 
As the long-run covariance consists of {infinitely many} lagged auto-covariance functions, we consider the kernel estimator of $LC_X(t,s,\tilde{t},\tilde{s})$, defined as
$$\widehat{LC}_X(t,s,\tilde{t},\tilde{s})=\sum_{h=-(N-1)}^{N-1}W\left(\frac{h}{\ell}\right)\widehat{\mathcal{C}}_{X,h}(t,s,\tilde{t},\tilde{s}),$$
where $W(\cdot)$ is a symmetric weight function satisfying the following assumptions.
\begin{assumption}
\label{a4}
$W(0)=1$, $0\le W(\cdot)\le1$, $W(u)=W(-u)$, $W(u)=0$ if $|u|>1$, and the bandwidth $\ell$ satisfies $\ell=N^{1/\kappa}$, where $\kappa>2$. 
\end{assumption}
\begin{assumption}
\label{a5}
There exist $\alpha>0,\beta>1$ and $c_0,c_1>0$, so that $c^{-1}_0h^{-\beta}\le\vertiii{\mathcal{C}_{X,h}}_\mathcal{S}\le c_0h^{-\beta}$, and $c^{-1}_1|u|^\alpha \le 1-W(u)\le c_1|u|^{\alpha}$.
\end{assumption}}
\begin{remark}
$\beta>1$ assures $LC_X(t,s,\tilde{t},\tilde{s})<\infty$, and $\alpha>0$ assures that the influence of auto-covariance decays as lag increases. As a special case, when the functions are independent, $W(0)=1$, and $W(u)=0,u\ne0$.
\end{remark}
 We estimate the  covariance and auto-covariance functions with the entire sequence as follows,
\begin{align*}
\widehat{\mathcal{C}}_{X,h}(t,s,\tilde{t},\tilde{s})&=\frac{1}{N-h}\sum_{i=1}^{N-h}\{X_i(t,s)-\bar{X}^*_i(t,s)\}\{X_{i+h}(\tilde{t},\tilde{s})-\bar{X}_{i+h}^*(\tilde{t},\tilde{s})\},\qquad h\ge0,\\
\widehat{\mathcal{C}}_{X,h}(t,s,\tilde{t},\tilde{s})&=\frac{1}{N+h}\sum_{i=1-h}^{N}\{X_i(t,s)-\bar{X}^*_i(t,s)\}\{X_{i+h}(\tilde{t},\tilde{s})-\bar{X}_{i+h}^*(\tilde{t},\tilde{s})\},\qquad h<0,
\end{align*}
where
\begin{equation*}
\bar{X}^*_i(t,s)=\left\{
\begin{aligned}
&\frac{1}{\hat{k}^*_N}\sum_{j=1}^{\hat{k}_N^*}X_j(t,s),\qquad 1\le i\le\hat{k}^*_N,\\
&\frac{1}{N-\hat{k}^*_N}\sum_{j=\hat{k}_N^*+1}^{N}X_j(t,s),\qquad \hat{k}^*_N+1\le i\le N,
\end{aligned}
\right.
\end{equation*}
and $\hat{k}^*_N=[N\hat{\theta}_N^*]$. As a side note, it can be shown that $\widehat{\mathcal{C}}_{X,h}(t,s,\tilde{t},\tilde{s})$ is an asymptotically unbiased estimator of ${\mathcal{C}}_{X,h}(t,s,\tilde{t},\tilde{s})$ under Assumption~\ref{a4}.
The estimated eigenvalues $\{\hat{\rho}_d\colon d\in\mathbb{N}\}$ are obtained by solving the following equation
\begin{equation*}
\int \widehat{LC}_X(t,s,\tilde{t},\tilde{s})\hat{\psi}_d(\tilde{t},\tilde{s})\tilde{t}d\tilde{s}=\hat{\rho}_d\hat{\psi}_d(t,s),
\end{equation*}
The following theorem, which holds under both $H_0$ and $H_a$, presents the convergence rate of the estimated long-run covariance $\widehat{LC}_X(t,s,\tilde{t},\tilde{s})$.
{
\begin{theorem}
\label{th6}
Under Assumption \ref{a1}---\ref{a5},  with arbitrary small $\epsilon>0$
\begin{align*}
\vertiii{\widehat{LC}_{X}-LC_X}_\mathcal{S}&\le O_p(1)N^{\max\{1/\kappa-1/2,-(\beta-1)/\kappa\}}\\
&\vee\left\{
\begin{array}{rcl}
N^{-(\beta-1)/\kappa}, & & \mbox{if}\ {\alpha-\beta>-1}.\\
N^{-(\beta-1)/\kappa+\epsilon},& & \mbox{if}\ {\alpha-\beta=-1}.\\
N^{-\alpha/\kappa}, & & \mbox{if}\ {\alpha-\beta<-1}.
\end{array} \right.
\end{align*}
\end{theorem}
By Corollary 1.6 in Gohberg et al.~(1990), $|{\hat{\rho}_d-\rho_d}|\le\vertiii{\widehat{LC}_{X}-LC_X}_\mathcal{S}$, then the  corollary below follows.
\begin{corollary}
\label{cr}
Under Assumption \ref{a1}---\ref{a5}, with arbitrary small $\epsilon>0$
\begin{align*}
\max\limits_{d\ge1}|{\hat{\rho}_d-\rho_d}|&\le O_p(1)N^{\max\{1/\kappa-1/2,-(\beta-1)/\kappa\}}\\
&\vee\left\{
\begin{array}{rcl}
N^{-(\beta-1)/\kappa}, & & \mbox{if}\ {\alpha-\beta>-1}.\\
N^{-(\beta-1)/\kappa+\epsilon},& & \mbox{if}\ {\alpha-\beta=-1}.\\
N^{-\alpha/\kappa}, & & \mbox{if}\ {\alpha-\beta<-1}.
\end{array} \right.
\end{align*}
\end{corollary}
\begin{remark}
As the functions are independent, only the covariance function is incorporated. By Theorem 3.1 in H\"ormann and Kokoszka (2010), in this special case, $\vertiii{\widehat{LC}_{X}-LC_X}_\mathcal{S}=O_p(N^{-1/2})$, and consequently, $\max\limits_{d\ge1}|{\hat{\rho}_d-\rho_d}|=O_p(N^{-1/2})$. 
\end{remark}

Note that, in practice, only finite number of eigenvalues $\rho_d$'s are estimated and incorporated, say, the truncated estimated null distribution $\sum_{d=1}^{D_\rho}\hat{\rho}_dB_d(\theta)$ is used to find the critical value. {Now we show how the selection of $D_\rho$ influences the estimation of the null distribution.}
Assuming ${D_\rho}$ eigenvalues are included and letting $\|\cdot\|$ signify the $l^2$-norm, it follows that
\begin{align*}
\left\|\sum_{d=1}^\infty\rho_dB^2_{d}(\theta)-\sum_{d=1}^{D_\rho}\hat{\rho}_dB^2_{d}(\theta)\right\|&\le\bigg\|\sum_{d=1}^{D_\rho}(\rho_d-\hat{\rho}_d){B}^2_{d}(\theta)\bigg\|+\bigg\|\sum_{d={D_\rho}+1}^\infty\rho_d{B}^2_{d}(\theta)\bigg\|.
\end{align*}
Suppose that there exists a constant $C$ and $\gamma>1$, such that $C^{-1}d^{-\gamma}\le\rho_d\le Cd^{-\gamma}$, then by the {triangle} inequality, 
\begin{align*}
\bigg\|\sum_{d=1}^{D_\rho}(\rho_d-\hat{\rho}_d){B}^2_{d}(\theta)\bigg\|&\le\sum_{d=1}^{D_\rho} |\rho_d-\hat{\rho}_d| \|{B}^2_{d}(\theta)\|\le O_p(1){D_\rho}N^{-\eta},\\
\bigg\|\sum_{d={D_\rho}+1}^\infty\rho_d{B}^2_{d}(\theta)\bigg\|&\le\sum_{d={D_\rho}+1}^\infty\rho_d\|{B}^2_{d}(\theta)\|\le O_p(1){D_\rho}^{-\gamma+1},
\end{align*}
{where $\eta$ is given in Corollary~\ref{cr} through $\alpha,\beta,\kappa$.} Assuming $D_\rho=N^{x}$ and $0<x<\eta$ assuring convergence, we derive the convergence rate of the approximation error as follows
$$\left\|\sum\limits_{d=1}^\infty\rho_dB^2_{d}(\theta)-\sum\limits_{d=1}^{D_\rho}\hat{\rho}_dB^2_{d}(\theta)\right\|\le O_p(1)\max\{N^{x-\eta},N^{-(\gamma-1)x}\}.$$ 

\begin{remark}
{When the functions consist of oscillations over a wide range of frequency (e.g., brain signals), the eigenvalue $\rho_d$ typically decays slow and a large amount of eigenvalues need to be incorporated. As the sample size is small, this could lead to overly large estimation error and reduced detection power. To solve this problem, we propose to filter the functions into different frequency bands and detect the change points in different bands separately.}  
\end{remark}

\subsection{Estimation of $\rho_d$'s}
\label{s3.3}
{Ramsay and Silverman (2004) developed a dimension reduction approach for the estimation of eigen-elements of binary covariance functions,} and we adjust and extend the procedure to estimate $\{\rho_d\colon d\in\mathbb{N}\}$.
We propose to represent $\{X_i(t,s)\colon i\in\mathbb{N}\}$ by a series of common basis functions. Given $\{\phi_d(t)\colon d\in\mathbb{N}\}$ being the common bases of $L^2(\mathcal{T})$, $\{\phi_d(t)\phi_{d'}(s)\colon d,d'\in\mathbb{N}\}$ are then the common bases of $L^2(\mathcal{T}\times\mathcal{T})$. If a function $f(t,s)$ is a symmetric function, we have $$\llangle f(t,s),\phi_d(t)\phi_{d'}(s)\rrangle=\llangle f(t,s),\phi_{d'}(t)\phi_d(s)\rrangle$$ for any pair of $d,d'$, therefore we can construct the following bases for bivariate symmetric functions,
\begin{align*}
\{\Phi_d(t,s)\colon d\in\mathbb{N}\}&\\
&\hspace{-3cm}=\left\{\phi_{d_1}(t)\phi_{d_2}(s)+\phi_{d_2}(t)\phi_{d_1}(s)\colon d_1\ne d_2\in\mathbb{N}_+\right\}\cup\left\{\phi_{d'}(t)\phi_{d'}(s)\colon d'\in\mathbb{N}_+\right\},
\end{align*}
Suppose the function $\widehat{Z}_i(t,s)=X_i(t,s)-\bar{X}_i^*(t,s)$ has the following basis approximation
\begin{equation*}
\label{b1}
\widehat{Z}_{i,J}(t,s)=\sum_{j=1}^{J}c_{ij}\Phi_{j}(t,s).
\end{equation*}
It is assumed that $J$ is selected such that the above $J$-dimensional approximation is close to the original functions. Specifically, for some tolerance error of approximation $e$, $J$ is the smallest number satisfying that $\sum_i\vertiii{\widehat{Z}_{i,J}-\widehat{Z}_i}^2/N$ is less than $e$. 
Define ${\Xi}_i=(c_{i,1},\ldots,c_{i,J})'$, $\bm{C}_{a:b}=({\Xi}_a,\ldots,{\Xi}_b)'$ where $(a<b)$, and $\bm{\Phi}(t,s)=(\Phi_1(t,s),\ldots,\Phi_J(t,s))'$. 
We represent $\widehat{LC}_X(t,s,\tilde{t},\tilde{s})$ in the following matrix form
\begin{align*}
\widehat{LC}_X(t,s,\tilde{t},\tilde{s})\approx&\sum_{h\ge0}W\left(\frac{h}{\ell}\right)\frac{1}{N-h}\bm{\Phi}(t,s)'\bm{C}'_{1:(N-h)}\bm{C}_{(h+1):N}\bm{\Phi}(\tilde{t},\tilde{s})\\
&+\sum_{h<0}W\left(\frac{h}{\ell}\right)\frac{1}{N+h}\bm{\Phi}(t,s)'\bm{C}'_{(1-h):N}\bm{C}_{1:(N+h)}\bm{\Phi}(\tilde{t},\tilde{s}).
\end{align*}
Now suppose that $\hat{\psi}(t,s)$ has the following basis approximation
$$\hat{\psi}(t,s)\approx\sum_{j=1}^Jb_{j}\Phi_j(t,s)=\bm{\Phi}(t,s)'{\bm{b}},$$
where $\bm{b}=(b_{1},\ldots,b_{J})'$, and this yields
\begin{align*}
\int\widehat{LC}_X(t,s,\tilde{t},\tilde{s}){\psi}(\tilde{t},\tilde{s})d\tilde{t}d\tilde{s}&\approx\int\bm{\Phi}(t,s)'\bm{\Sigma}_C\bm{\Phi}(\tilde{t},\tilde{s})\bm{\Phi}(\tilde{t},\tilde{s})'\bm{b}d\tilde{t}d\tilde{s}\\
&=\bm{\Phi}(t,s)'\bm{\Sigma}_CG\bm{b}\\
&\approx\rho\bm{\Phi}(t,s)'\bm{b},
\end{align*}
where
$$\bm{\Sigma}_C=\sum_{h\ge0}W\left(\frac{h}{\ell}\right)\frac{1}{N-h}\bm{C}'_{1:(N-h)}\bm{C}_{(h+1):N}+\sum_{h<0}W\left(\frac{h}{\ell}\right)\frac{1}{N+h}\bm{C}'_{(1-h):N}\bm{C}_{1:(N+h)}.$$
where $G$ is a $J\times J$ matrix with elements $G_{ij}=\llangle\Phi_i,\Phi_j\rrangle$. This equation holds for arbitrary $t$ and $s$, thus we have the following approximated eigen-equation
$${G}^{1/2}\bm{\Sigma}_C{G}^{1/2}\bm{u}=\rho \bm{u},$$ 
where $\bm{u}={G}^{1/2}\bm{b}$. We propose to solve this eigen-equation to obtain $\{\hat{\rho}_d\colon d\in\mathbb{N}\}$. 
\subsection{Asymptotic properties of the estimated change point}
We now develop the asymptotic properties of the estimated change point. Denote $k^*=[N\theta^*]$, where $\theta^*$ is fixed and unknown. First we shall show $\hat{\theta}^*_N\overset{p}\to\theta^*$. Define $\Delta_c(t,s)=C^{(1)}(t,s)-C^{(2)}(t,s)$, and for $\theta\in[0,1]$,
\begin{equation*}
h(\theta,t,s)=\left\{
\begin{aligned}
\theta(1-\theta^*)\Delta_c(t,s),\qquad 0\le\theta\le\theta^*,\\
(1-\theta)\theta^*\Delta_c(t,s),\qquad \theta^*<\theta\le1,
\end{aligned}
\right.
\end{equation*}

\begin{theorem}
\label{th5}
Under Assumption 1, if $\Delta_c\ne0$, then
$$\sup_{0\le\theta\le1}\left\{N^{-1}T_N(\theta)-\int\int h^2(\theta,t,s)dtds\right\}\overset{p}\to0.$$
\end{theorem}
We demonstrate the consistency of $\hat{\theta}_N^*$ in the following corollary.
\begin{corollary}
\label{co1}
Under Assumption 1, if $\Delta_c\ne0$, then $\hat{\theta}_N^*\overset{p}\to\theta^*$.
\end{corollary}

This corollary can be easily obtained from Theorem~\ref{th5}. The estimated scaled change point is the maximizer of $T_N(\theta)/N$. Evidently, the unique maximizer of $\int\int h^2(\theta,t,s)dtds$ is $\theta^*$, thus $\hat{\theta}^*_N\overset{p}\to\theta^*$. 

To discuss the asymptotic distribution of the estimated unscaled change point, we define
\begin{equation*}
U(k)=\left\{
\begin{aligned}
\left\{(1-\theta^*)\vertiii{\Delta_c}^2+\frac{1}{k}\sum_{i=k^*+k}^{k^*}\llangle \Delta_c,X_i-C^{(1)}\rrangle\right\}k,\qquad k<0,\\
0,\qquad k=0,\\
\left\{-\theta^*\vertiii{\Delta_c}^2-\frac{1}{k}\sum_{i=k^*+1}^{k^*+k}\llangle \Delta_c,X_i-C^{(2)}\rrangle\right\}k,\qquad k>0,
\end{aligned}
\right.
\end{equation*}
The difference between the estimated unscaled change point $\hat{k}_N^*$ and the true unscaled change point $k^*$ asymptotically converges to the smallest maximizer of the function $U(k)$ in distribution, which is illustrated in Theorem~\ref{th3}.

\begin{theorem}
\label{th3}
Under Assumption 1 and {the assumption that the distribution of $Y_i(t)$ is continuous over $t\in\mathcal{T}$}, if $\Delta_c\ne 0$, then
$$\hat{k}^*_N-k^*\overset{\mathcal{D}}\to\min\left\{k\colon U(k)=\sup_{\tilde{k}\in\mathbb{N}}U(\tilde{k})\right\},\qquad \text{as}\ N\to\infty.$$
\end{theorem}

\begin{remark}
{Under the assumption of continuous distribution, the probability that $U(k)$ has more than one maximizer is zero.}
The asymptotic distribution of $\hat{k}^*_N$ is influenced by two factors: (1.)~discrepancy $C^{(1)}(t,s)-C^{(2)}(t,s)$, and (2.)~variation of $\{X_i(t,s)\colon i\in\mathbb{N}\}$ and the alignment between $X_i(t,s)-E\{X_i(t,s)\}$ and $\Delta_c(t,s)$. As a special case, if $X_i(t,s)-E\{X_i(t,s)\}$ is orthogonal to $\Delta_c(t,s)$ for any $i$, then the unscaled estimated change point is always consistent with the true one. 
\end{remark}

\section{{Simulations}}
\label{s4}
\subsection{Settings}
To study the finite sample behaviors of the change point estimator, we simulated two groups of functions with the same sample size and different covariance functions. The two groups of functions were concatenated as a functional sequence with structural break in the mid-point. The functional sequences were simulated either from an i.i.d.~process or a FAR(1) process. We selected the 2-nd to the 9-th Fourier basis over the unit interval $[0,1]$, denoted by $\nu_1(t),\ldots,\nu_{8}(t)$, to generate the functions. In other words, we simulated functions in the $\delta$-frequency band (1-4 Hertz). The curves were then generated by the basis expansion
$$Y^{(g)}_i(t)=\sum_{d=1}^{8}\xi^{(g)}_{i,d}\nu_d(t)+e_i(t),\qquad g=1,2.$$
If the generating process is an i.i.d.~process, $\{\xi^{(g)}_{i,d}\colon d=1,\ldots,8\}$ are independent normal random variables with standard deviation ${\bm\sigma}_1$ and ${\bm\sigma}_2$ for group 1 and group 2 respectively, and if the generating process is a FAR(1) process, $\{\xi^{(g)}_{i,d}\colon d=1,\ldots,8\}$ satisfies the recursive equation $\xi^{(g)}_{i,d}=0.5\xi^{(g)}_{i-1,d}+\epsilon^{(g)}_{i,d}$, where $\{\epsilon^{(g)}_{i,d}\colon d=1,\ldots,8\}$ are independent normal random variables with standard deviation ${\bm\sigma}_1$ and ${\bm\sigma}_2$ for group 1 and group 2 respectively. The bandwidth is $N^{1/4}$ for the dependent case, and 1 for the i.i.d.~case, and $W(u)=1$, $|u|\le1$. 

Denote $\bm{1}_p$ to be the $p$-dimensional row vector with all elements being $1$. We considered three different settings. {When there is no change point, 
\begin{itemize}
\item Setting 1: ${\bm\sigma}_1={\bm\sigma}_2=(\bm{1}_4,0.5\bm{1}_4)$;
\item Setting 2: ${\bm\sigma}_1={\bm\sigma}_2=(\bm{1}_2,0.5\bm{1}_2,\bm{1}_2,0.5\bm{1}_2)$;
\item Setting 3: ${\bm\sigma}_1={\bm\sigma}_2=(1,0.5)\otimes\bm{1}_4$,
\end{itemize}}
and when there is a change point in the middle,
\begin{itemize}
\item Setting 1: ${\bm\sigma}_1=(\bm{1}_4,0.5\bm{1}_4)$, ${\bm\sigma}_2=(0.5\bm{1}_4,\bm{1}_4)$;
\item Setting 2: ${\bm\sigma}_1=(\bm{1}_2,0.5\bm{1}_2,\bm{1}_2,0.5\bm{1}_2)$, ${\bm\sigma}_2=(0.5\bm{1}_2,\bm{1}_2,0.5\bm{1}_2,\bm{1}_2)$;
\item Setting 3: ${\bm\sigma}_1=(1,0.5)\otimes\bm{1}_4$, ${\bm\sigma}_2=(0.5,1)\otimes\bm{1}_4$.
\end{itemize}
When a change point exists, in the first two settings, the discrepancy between the covariance functions comes from the difference of spectral distribution, and functions in group 1 contain lower frequency oscillations. In Setting 3, the two groups have the same spectrum but different phase distribution. $\{e_i(t)\colon i\in\mathbb{N}\}$ are i.i.d.~random error functions satisfying 
$$e_i(t)=\sum_{d=1}^{8}\xi^{(e)}_{i,d}\nu_d(t),$$ 
where $\{\xi^{(e)}_{i,d}\colon d=1,\ldots,8\}$ are independent normal random variables with mean zero and standard deviation ${\bm\sigma}_e=\{\sigma/d\colon d=1,\ldots,8\}$. We took into account the influence of random error on the detection performance by setting different values to $\sigma$, say, $\sigma=3,4,5,6$.  A large value of $\sigma$ indicates low signal-noise ratio. For each setting, we simulated $150$ or $300$ curves for each group. The simulation runs were repeated {1000} times for different values of $\sigma$.
 
\subsection{Size and power} 
 
{Under the settings without a change point, we calculated the empirical size at the nominal level $\alpha=0.05$.
Under the settings with a change point, we obtained the empirical power at level $\alpha=0.05$. We used the R package ``$sde$'' to obtain the numerical 95\% quantile of the null distribution. Table~\ref{t1}---\ref{t4} display the empirical sizes and powers under different settings, and Figure~\ref{boxiid} and \ref{boxfar} displays the Box-plots of $\hat{\theta}^*_N$. To estimate the non-pivotal null distribution, the number of eigenvalues were selected to be $N^x$, where $x$ takes value in $0.28, 0.36, 0.42, 0.5$, corresponding to $d_1,d_2,d_3,d_4$. 

\begin{table}[H]
	\centering
	\caption{Empirical sizes under different settings and dimensions (i.i.d.)}{
	\begin{tabular}{c p{0.2in}|p{0.28in}p{0.28in}p{0.28in}p{0.28in}|p{0.28in}p{0.28in}p{0.28in}p{0.28in}|p{0.28in}p{0.28in}p{0.28in}p{0.28in}}
	\hline 
         $\sigma$  & $N$ & \multicolumn{4}{c|}{\tabincell{c}{Setting 1\\$d_1$\hspace{0.57cm} $d_2$\hspace{0.57cm}  $d_3$ \hspace{0.57cm} $d_4$}} &\multicolumn{4}{c|}{\tabincell{c}{Setting 2\\$d_1$\hspace{0.57cm} $d_2$\hspace{0.57cm}  $d_3$ \hspace{0.57cm} $d_4$}}  & \multicolumn{4}{c}{\tabincell{c}{Setting 3\\$d_1$\hspace{0.57cm} $d_2$\hspace{0.57cm}  $d_3$ \hspace{0.57cm} $d_4$}} \\
         \hline 
         $3$ &150&0.12& 0.09& 0.06& 0.06&0.10& 0.09& 0.08& 0.07&0.10& 0.10& 0.08& 0.06\\
                &300&0.09& 0.08& 0.06& 0.05&0.08& 0.08& 0.06& 0.04&0.10& 0.08& 0.06& 0.05\\
                \hline 
         $4$ &150&0.09& 0.07& 0.07& 0.06&0.10& 0.10& 0.08& 0.05&0.10& 0.06& 0.08& 0.06 \\
	        &300&0.08& 0.06& 0.07& 0.06&0.07& 0.07& 0.05& 0.05&0.09& 0.06& 0.05& 0.05\\
	        \hline 
         $5$  &150&0.07& 0.08& 0.08& 0.06&0.08& 0.08& 0.07& 0.04&0.07& 0.08& 0.08& 0.07\\
	        &300&0.07& 0.07 &0.06 &0.06&0.08 &0.07 &0.06 &0.06&0.10 &0.08 &0.06 &0.07\\
	        \hline 
	 $6$  &150&0.09& 0.08& 0.06& 0.04&0.08 &0.07& 0.06 &0.06&0.09& 0.09 &0.05& 0.06\\
	        &300&0.07 &0.05& 0.06 &0.04&0.07& 0.07 &0.07& 0.06&0.06 &0.06& 0.05 &0.06\\
	 \hline 
	\end{tabular}
	}
	\label{t1}
\end{table}

\begin{table}[H]
	\centering
	\caption{Empirical powers under different settings and dimensions (i.i.d.)}{
	\begin{tabular}{c p{0.2in}|p{0.28in}p{0.28in}p{0.28in}p{0.28in}|p{0.28in}p{0.28in}p{0.28in}p{0.28in}|p{0.28in}p{0.28in}p{0.28in}p{0.28in}}
	\hline 
         $\sigma$  & $N$ & \multicolumn{4}{c|}{\tabincell{c}{Setting 1\\$d_1$\hspace{0.57cm} $d_2$\hspace{0.57cm}  $d_3$ \hspace{0.57cm} $d_4$}} &\multicolumn{4}{c|}{\tabincell{c}{Setting 2\\$d_1$\hspace{0.57cm} $d_2$\hspace{0.57cm}  $d_3$ \hspace{0.57cm} $d_4$}}  & \multicolumn{4}{c}{\tabincell{c}{Setting 3\\$d_1$\hspace{0.57cm} $d_2$\hspace{0.57cm}  $d_3$ \hspace{0.57cm} $d_4$}} \\
         \hline 
         $3$ &150&1.00& 1.00& 1.00& 1.00&1.00& 1.00& 1.00& 1.00&1.00& 1.00& 1.00& 1.00\\
                &300&1.00& 1.00& 1.00& 1.00&1.00& 1.00& 1.00& 1.00&1.00& 1.00& 1.00& 1.00\\
                \hline 
         $4$ &150&0.93& 0.91& 0.89& 0.83&0.93& 0.90& 0.89& 0.87&0.93& 0.91& 0.91& 0.86 \\
	        &300&1.00& 1.00& 1.00& 1.00&1.00& 1.00& 1.00& 1.00&1.00& 1.00& 1.00& 1.00\\
	        \hline 
         $5$  &150&0.64& 0.63& 0.60& 0.57&0.66& 0.63& 0.56& 0.56&0.68& 0.62& 0.60& 0.56\\
	        &300&1.00& 0.99& 0.99& 0.99&1.00& 1.00& 0.99 &0.98&1.00 &0.99& 0.99 &0.99\\
	        \hline 
	 $6$  &150&0.42& 0.36& 0.36& 0.35&0.43& 0.40& 0.38& 0.28&0.42& 0.37 &0.38& 0.34\\
	        &300&0.92 &0.87 &0.82 &0.82&0.92 &0.87 &0.85 &0.81&0.92 &0.86& 0.82 &0.81\\
	 \hline 
	\end{tabular}
	}
	\label{t2}
\end{table}

\begin{table}[H]
	\centering
	\caption{Empirical sizes under different settings and dimensions (FAR(1))}{
	\begin{tabular}{c p{0.2in}|p{0.28in}p{0.28in}p{0.28in}p{0.28in}|p{0.28in}p{0.28in}p{0.28in}p{0.28in}|p{0.28in}p{0.28in}p{0.28in}p{0.28in}}
	\hline 
         $\sigma$  & $N$ & \multicolumn{4}{c|}{\tabincell{c}{Setting 1\\$d_1$\hspace{0.57cm} $d_2$\hspace{0.57cm}  $d_3$ \hspace{0.57cm} $d_4$}} &\multicolumn{4}{c|}{\tabincell{c}{Setting 2\\$d_1$\hspace{0.57cm} $d_2$\hspace{0.57cm}  $d_3$ \hspace{0.57cm} $d_4$}}  & \multicolumn{4}{c}{\tabincell{c}{Setting 3\\$d_1$\hspace{0.57cm} $d_2$\hspace{0.57cm}  $d_3$ \hspace{0.57cm} $d_4$}} \\
         \hline 
         $3$ &150&0.16& 0.11& 0.08& 0.05&0.15& 0.11& 0.06& 0.07&0.17& 0.13 &0.09 &0.07\\
                &300&0.11& 0.07& 0.08 &0.06&0.12& 0.09& 0.05& 0.06&0.15& 0.09& 0.06& 0.06\\
                \hline 
         $4$ &150&0.12& 0.10& 0.08& 0.07&0.13& 0.10& 0.08& 0.07&0.13& 0.10& 0.07& 0.07 \\
	        &300&0.11 &0.08& 0.05& 0.06&0.12& 0.08& 0.06& 0.05&0.14 &0.07& 0.07& 0.06\\
	        \hline 
         $5$  &150&0.11& 0.11& 0.07& 0.07&0.14& 0.10& 0.08& 0.07&0.12& 0.10& 0.09& 0.08\\
	        &300&0.12& 0.07& 0.08& 0.07&0.11& 0.08& 0.06& 0.06&0.11& 0.07& 0.06& 0.06\\
	        \hline 
	 $6$  &150&0.12& 0.11& 0.07& 0.07&0.12& 0.07& 0.08& 0.06&0.11& 0.09& 0.08& 0.06\\
	        &300&0.09& 0.07& 0.07& 0.06&0.11& 0.07& 0.07& 0.06&0.12& 0.08& 0.08& 0.06\\
	 \hline 
	\end{tabular}
	}
	\label{t3}
\end{table}

\begin{table}[H]
	\centering
	\caption{Empirical powers under different settings and dimensions (FAR(1))}{
	\begin{tabular}{c p{0.2in}|p{0.28in}p{0.28in}p{0.28in}p{0.28in}|p{0.28in}p{0.28in}p{0.28in}p{0.28in}|p{0.28in}p{0.28in}p{0.28in}p{0.28in}}
	\hline 
         $\sigma$  & $N$ & \multicolumn{4}{c|}{\tabincell{c}{Setting 1\\$d_1$\hspace{0.57cm} $d_2$\hspace{0.57cm}  $d_3$ \hspace{0.57cm} $d_4$}} &\multicolumn{4}{c|}{\tabincell{c}{Setting 2\\$d_1$\hspace{0.57cm} $d_2$\hspace{0.57cm}  $d_3$ \hspace{0.57cm} $d_4$}}  & \multicolumn{4}{c}{\tabincell{c}{Setting 3\\$d_1$\hspace{0.57cm} $d_2$\hspace{0.57cm}  $d_3$ \hspace{0.57cm} $d_4$}} \\
         \hline 
         $3$ &150&1.00& 1.00& 1.00& 1.00&1.00& 1.00& 1.00& 1.00&1.00& 1.00& 1.00& 1.00\\
                &300&1.00& 1.00& 1.00& 1.00&1.00& 1.00& 1.00& 1.00&1.00& 1.00& 1.00& 1.00\\
                \hline 
         $4$ &150&1.00& 0.98& 0.98& 0.97&0.99 &0.98& 0.98& 0.98&0.99& 0.98& 0.98& 0.97 \\
	        &300&1.00& 1.00& 1.00& 1.00&1.00& 1.00& 1.00& 1.00&1.00& 1.00& 1.00& 1.00\\
	        \hline 
         $5$  &150&0.93& 0.91& 0.86& 0.84&0.94& 0.90& 0.86& 0.80&0.93& 0.90& 0.87& 0.84\\
	        &300&1.00& 1.00& 1.00& 1.00&1.00& 1.00& 1.00 &1.00&1.00& 1.00& 1.00& 1.00\\
	        \hline 
	 $6$  &150&0.81& 0.71& 0.67& 0.59&0.81& 0.73& 0.65& 0.65&0.80& 0.71& 0.65& 0.64\\
	        &300&1.00& 1.00& 0.99& 0.99&1.00& 1.00& 1.00& 0.99&1.00& 1.00& 1.00& 0.99\\
	 \hline 
	\end{tabular}
	}
	\label{t4}
\end{table}
From Table \ref{t1}---\ref{t4}, we found that the empirical size of the proposed method is typically robust to the selection of dimensions as long as a sufficient number of $\rho_d$ is selected, and the powers increase substantially as the sample size (number of curves) increases. Figure \ref{boxiid} and \ref{boxfar} show that the variance of $\hat{\theta}_N^*$ shrinks significantly as $\sigma$ decreases and sample size increases.

There is also one interesting point in Setting 3 that needs to be emphasized. The two important elements in frequency domain analysis are the 
spectra and phase (see Ombao and Pinto, 2021). As testing the structural break in brain signal recordings (e.g., ~EEG, LPF), we can check the spectrum function (see e.g., Schr\"oder and Ombao, 2019). However, under Setting 3, the spectrum function is the same for the entire sequence, and consequently the spectrum-based detection method does not work, but as our functional procedure incorporates intra-curve information, the structural break in phase can also be detected. This is one of the major advantages of our procedure. 
\begin{figure}[H]
\center
\caption{Box-plots of the detected change points (i.i.d.)}
\includegraphics[scale=0.5]{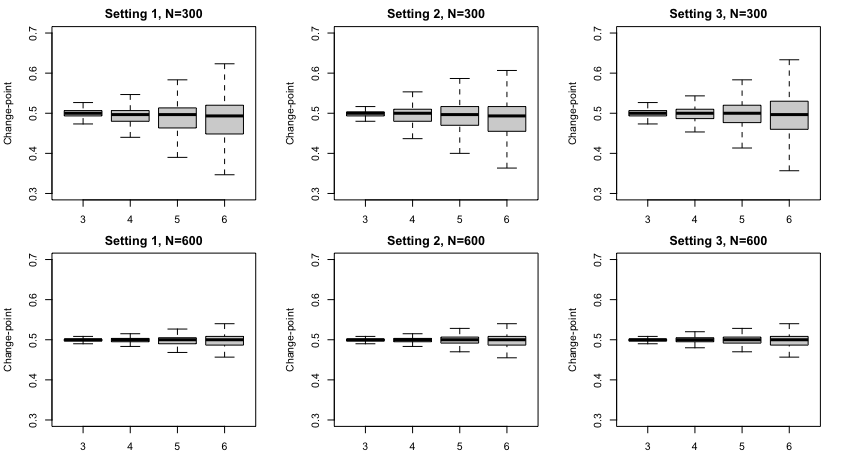}
\label{boxiid}
\end{figure}
\begin{figure}[H]
\center
\caption{Box-plots of the detected change points (FAR(1))}
\includegraphics[scale=0.5]{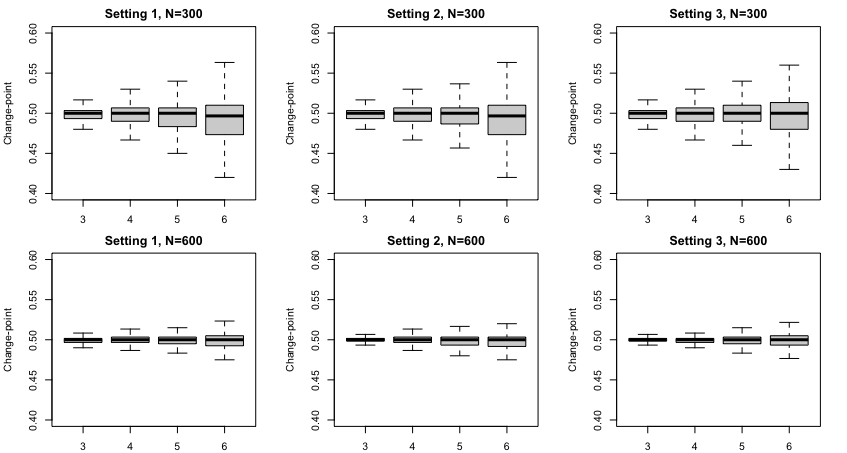}
\label{boxfar}
\end{figure}

\subsection{Comparison with other approaches}
\subsubsection{Bootstrap approach}
Since the null distribution is non-pivotal, Sharipov et al.~(2016) proposed a bootstrap approach to obtain the empirical critical value. In their approach, the entire sequence ($N$ functions) are segmented into $p$ blocks, where each block includes $\lfloor N/p\rfloor$ functions. The blocks are then resampled with replacement to form a new bootstrap sequence. The bootstrap repetition are repeated for multiple times, and the CUSUM are calculated for each repetition (Here, 1000 simulation runs for each setting and 1000 bootstrap repetitions for each simulation run). The critical value are then obtained from the bootstrap CUSUMs. Here we studied the performance of the bootstrap method on i.i.d.~sequences by simulation, and $p$ is set to be $N$. The empirical sizes and powers were obtained at nominal level $\alpha=0.05$. The results are displayed in Table~\ref{bs}. It is shown that the empirical size of the bootstrap approach is marginally lower than the nominal level, and meanwhile the empirical power is lower than the dimension reduction approach, especially under small sample size and large value of $\sigma$. 
\begin{table}[H]
	\centering
	\caption{Empirical sizes and powers under different setting (i.i.d.)}{
	\begin{tabular}{c p{0.2in}|cc|cc|cc}
	\hline 
         $\sigma$  & $N$ & \multicolumn{2}{c|}{\tabincell{c}{Setting 1\\Size  \hspace{0.14cm}  Power}} &\multicolumn{2}{c|}{\tabincell{c}{Setting 2\\Size  \hspace{0.14cm}  Power}}  & \multicolumn{2}{c}{\tabincell{c}{Setting 3\\Size  \hspace{0.14cm}  Power}} \\
         \hline 
         $3$ &150&0.046& 0.989& 0.035& 0.987 &0.033& 0.991\\
                &300&0.051& 1.000&  0.050& 1.000&0.044& 1.000\\
                \hline 
         $4$ &150&0.046& 0.770& 0.036& 0.749&0.039& 0.765 \\
	        &300&0.040& 1.000&  0.052& 1.000&0.037& 1.000\\
	        \hline 
         $5$  &150& 0.037& 0.414& 0.046& 0.429&0.047& 0.428\\
	        &300&0.056& 0.973&0.043& 0.985&0.036& 0.980\\
	        \hline 
	 $6$  &150&0.042& 0.241& 0.050& 0.233&0.047&0.256\\
	        &300&0.053 &0.800 &0.048 &0.791&0.044 &0.786\\
	 \hline 
	\end{tabular}
	}
	\label{bs}
\end{table}

\subsubsection{Weighted CUSUM}
Ordinary CUSUM statistics works well when the change point occurs in the middle of a sequence. However, as the change point is not located in the middle of a sequence, a weighted CUSUM can be considered, say, $$T_{w,N}(\theta)=T_N(\theta)/\{\theta(1-\theta)\}.$$ Since the weight $\theta(1-\theta)$ goes to zero as $\theta$ goes to 0 or 1, we cannot develop the consistency of $\hat{\theta}^*_N$ over the interval $(0,1)$. However, if the change point is bounded away from the boundaries, the consistency of $\hat{\theta}^*_N$ still holds. In the following, we assume there exists some arbitrary small $\epsilon>0$, so that $\theta^*\in[\epsilon,1-\epsilon]$. With the same argument of Theorem \ref{th1}, it can be obtained that under $H_0$,
$$T_{w,N}(\hat{\theta}^*_{N})\overset{\mathcal{D}}\to\sup_{\theta\in[\epsilon,1-\epsilon]}\frac{\sum_{d=1}^\infty\rho_dB^2_{d}(\theta)}{{{{\theta}}(1-{{\theta}})}}, \qquad N\to\infty.$$ 
Here, $\hat{\theta}^*_{N}=\inf\{\theta\colon T_{\omega,N}(\theta)=\sup_{\epsilon\le\theta'\le1-\epsilon}{T_{\omega,N}(\theta')}\}.$
We considered five change points (scaled) $\theta^*=0.1,0.2,0.3,0.4,0.5$,  and simulated 600 functions under different settings. We compared the location of the estimated change points obtained from the unweighted (ordinary) and weighted CUSUM statistics. The box-plots of $\hat{\theta}_N^*$ are displayed in Figure~\ref{boxweight}. 

\begin{figure}[ht]
\center
\caption{Box-plots of the detected change points obtained from ordinary and weighted CUSUM (i.i.d., nominal level $\alpha=0.05$, $N=600$)}
\includegraphics[scale=0.4]{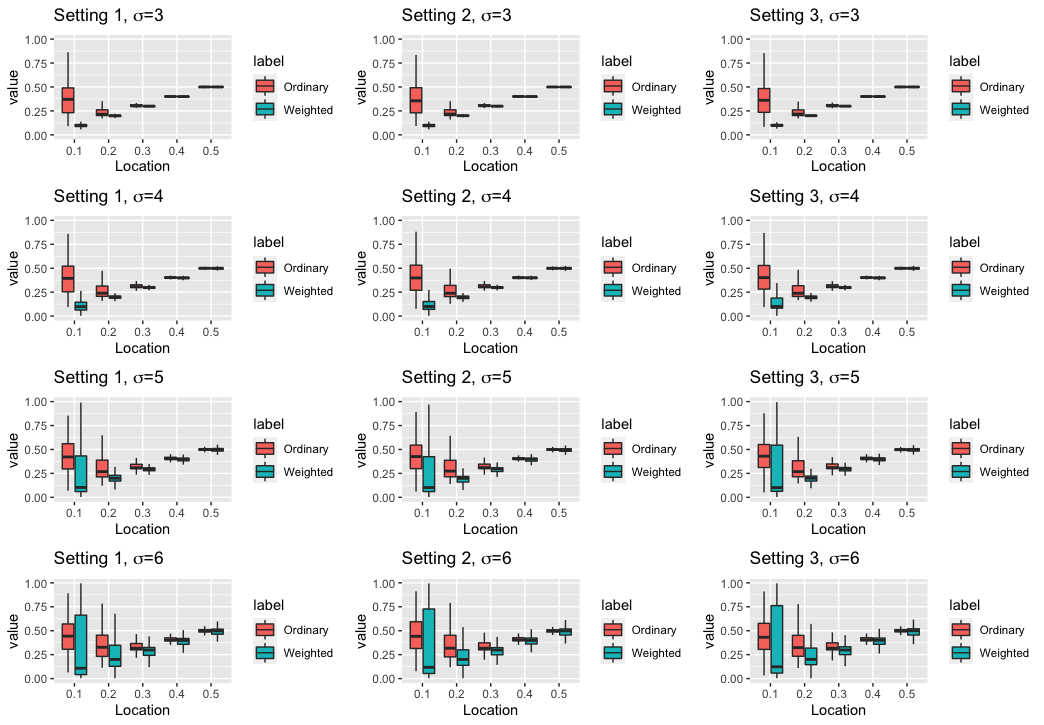}
\label{boxweight}
\end{figure}

It is noted that, as $\theta^*=0.4,0.5$, the ordinary CUSUM estimator is more robust. The superiority becomes pronounced as $\sigma$ increases. In addition, as $\theta^*=0.1$, the variance of the estimated change points by the weighted CUSUM is very large, and the ordinary CUSUM cannot detect the true change point. The reason is that the sample size over $[0,k^*]$ is small (60 curves). To solve this problem, we simulated 2400 curves in each setting to study the influence of sample size, and the box-plots are displayed in Figure \ref{boxweight2}. It is shown that the weighted CUSUM works well as the sample size is large enough, and is superior to ordinary CUSUM when detecting change points near the boundaries.

\begin{figure}[H]
\center
\caption{Box-plots of detected change points obtained from ordinary and weighted CUSUM (i.i.d., nominal level $\alpha=0.05$, $N=2400$)}
\includegraphics[scale=0.4]{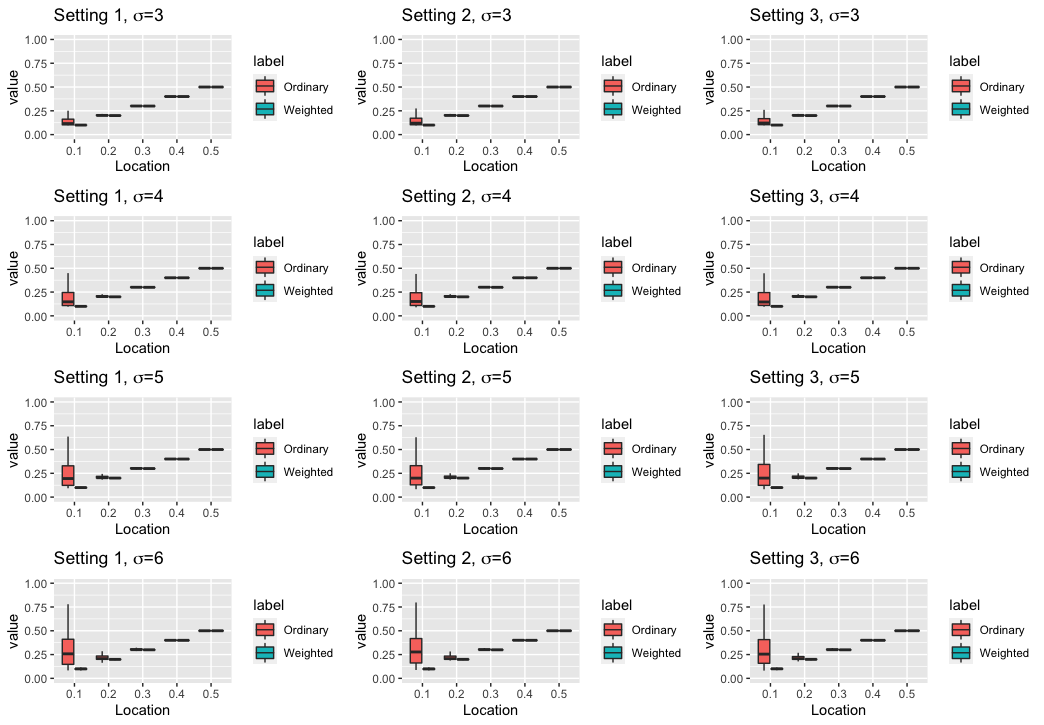}
\label{boxweight2}
\end{figure}

We summarize the findings as follows.
\begin{itemize}
\item[(1.)] When the true change point locates near the middle of a sequence, the unweighted CUSUM, $T_N(\theta)$, produces more robust detection results. 
\item[(2.)] When the true change point locates near the boundary of a sequence, the weighted CUSUM shows superiority over the unweighted CUSUM. Due to the chance variation near the boundaries, the variance of the estimated change point obtained from the weighted CUSUM can be overly large. One way to solve this problem is to increase the sample size to attenuate the effect of chance variation near the boundaries. 
\end{itemize}
}
\section{{Application to rat local field potentials}}
\label{s5}
The new method was applied to local field potential (LPF) trajectories of rat brain activity, collected from a stroke experiment reported in (Wann, 2017). Micro-tetrodes were inserted in 32 locations on the rat cortex from which LFPs were recorded at the rate of 1000 observations per second (Figure \ref{display}). 

\begin{figure}[H]
\center
\caption{Display of 32 micro-tetrodes}
\includegraphics[scale=0.3]{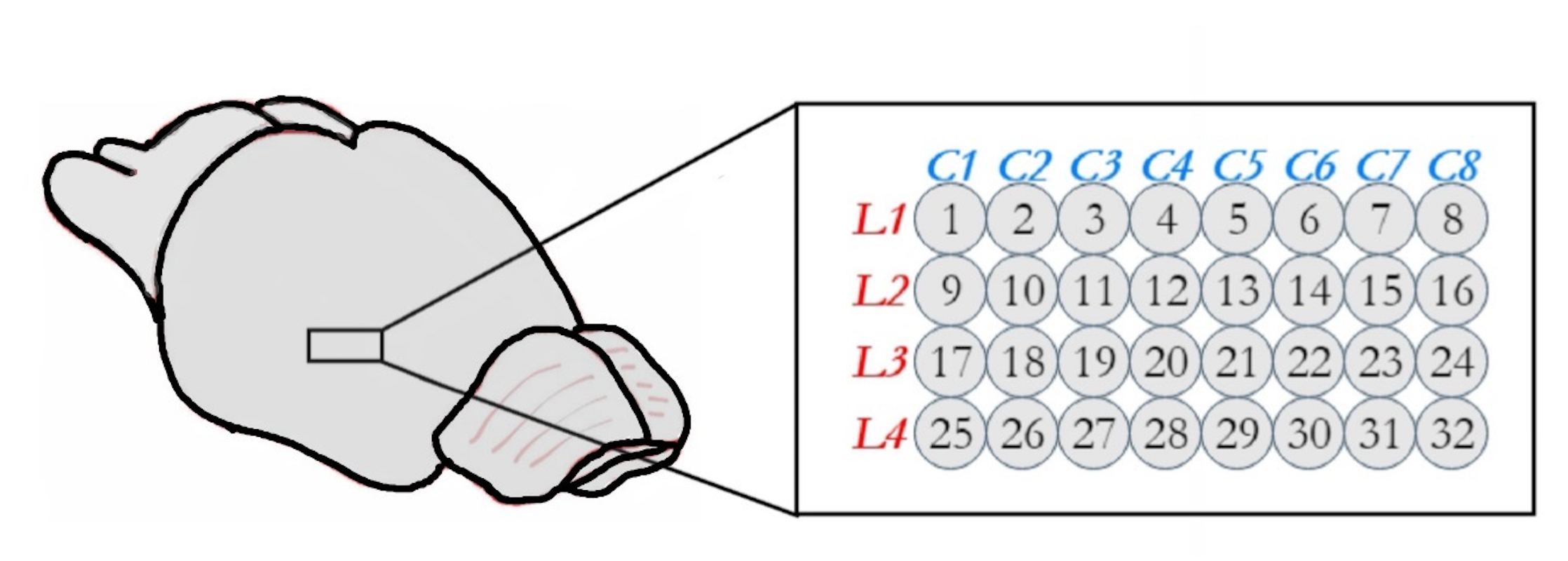}
\label{display}
\end{figure}

In our analysis, observations collected in one second is considered as an epoch. The data at hand  consists of 10 minute of recordings which leads to a total of 600 epochs. Midway in this period (at epoch 300), stroke was mechanically induced on the rat by clamping the medial cerebral artery. 
Here we considered the $\delta$-frequency band (0.5-4 Hertz), and smoothed the trajectory of each epoch with the first 9 Fourier bases specified as follows
\begin{equation*}
F_i(t)=\left\{
\begin{array}{cccl}
1, &t\in[0,1],  & & \text{if}\ i=1,\\
\sqrt{{2}}\cos(2\pi {k}t), & t\in[0,1], && \text{if}\ i=2k,\\
\sqrt{{2}}\sin(2\pi {k}t), & t\in[0,1],&&  \text{if}\ i=2k+1,
\end{array} \right. 
\end{equation*}
where $k=1,2,3,4$. If other frequency bands are of interest, the raw epoch trajectories can be smoothed with the Fourier bases in the corresponding frequency band. 

Irregular extremely large fluctuations may be observed after the occlusion of brain artery. To stabilize the variance of recordings, we applied the cubic root transformation on the LFP values, and outlier epochs were removed from each tetrode. Here, the outlier epochs for each tetrode are defined as those whose norm is beyond the interval $[Q_1 - 1.5\times \mbox{IQR},Q_3 + 1.5\times \mbox{IQR}]$, where $\mbox{IQR}=Q_3-Q_1$ and $Q_1, Q_3$ are the first and third quantile of the $l_2$-norm of the epoch trajectories. The pre-processed LFPs of the 32 tetrodes are displayed in Figure \ref{LFP1}, where the vertical dotted line marks the time of occlusion.

\begin{figure}[ht]
\center
\caption{Pre-processed LFPs of 32 micro-tetrodes ($\delta$-band)}
\includegraphics[scale=0.35]{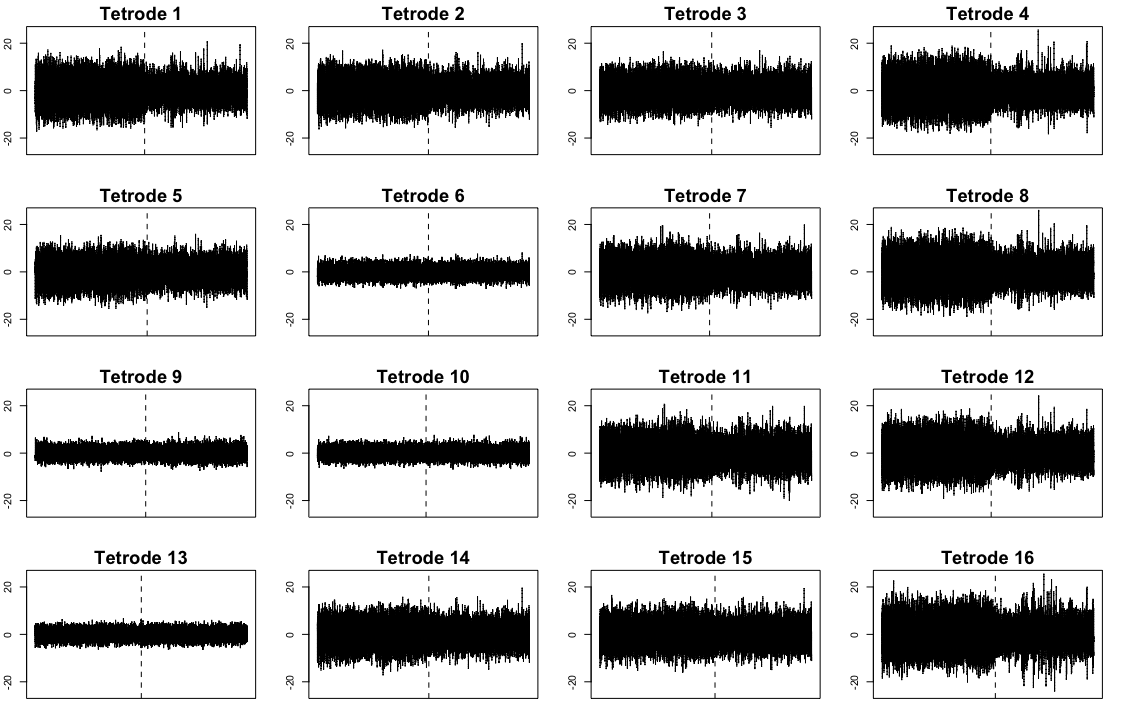}
\hbox{\hspace{0.35em}}\includegraphics[scale=0.35]{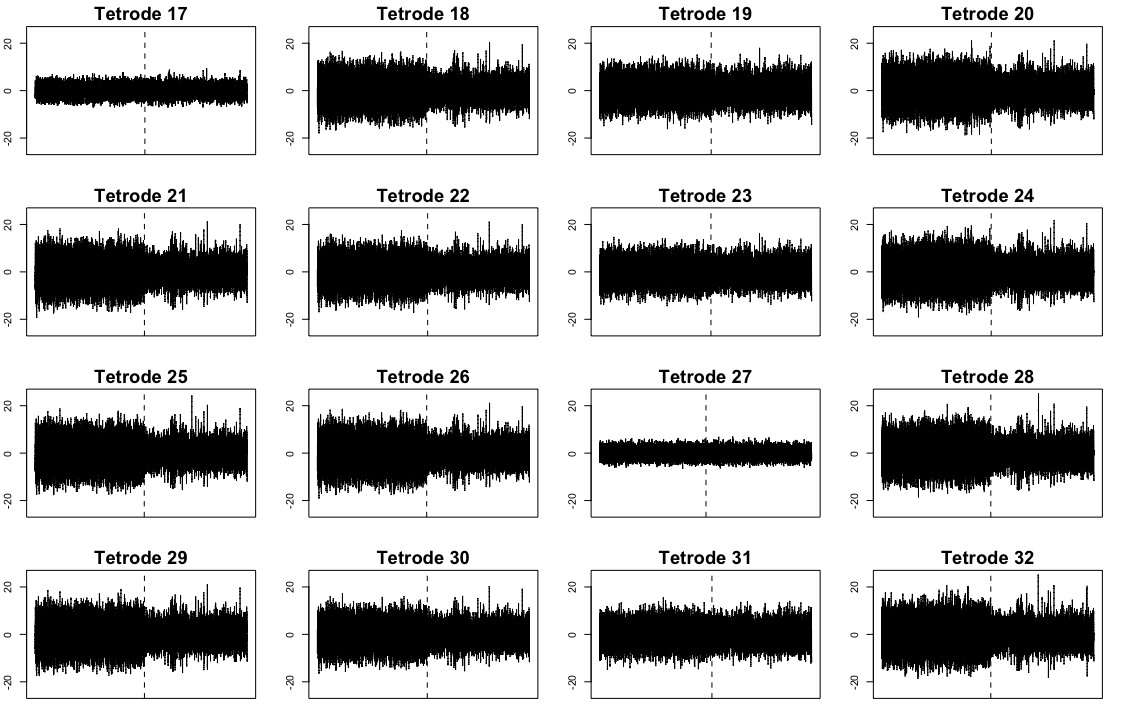}
\label{LFP1}
\end{figure}
We applied the detection procedure to each pre-processed sequence, and the estimated change points were tested significant at level 0.05 for most of the tetrodes, indicating pronounced structural change in the brain. Only tetrode 6 and 27 do not show significant structural break in the covariance function. Figure~\ref{f4} displays the difference between the estimated (scaled) change point and the time of occlusion (scaled) (each square represents one tetrode). $p<0.05$ means the $p$-value is below 0.05 and the $H_0$ is rejected, and $p>0.05$ means the other way. It shows that for most tetrodes, the estimated change point coincides with the true one.  For tetrodes 9,11,13, there is a substantial delay of the structural break after the artificial artery occlusion.

\begin{figure}[H]
\center
\caption{Difference between the estimated change point and the time of occlusion (scaled)}
\includegraphics[scale=0.74]{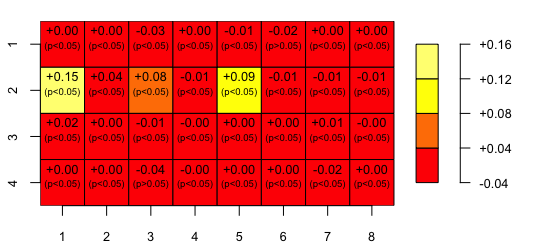}
\label{f4}
\end{figure}

\section{Conclusion and Future Work}
\label{s6}
In conclusion, we developed a procedure to identify the change point in the covariance function {of weakly dependent functional data.} The method is demonstrated to be useful when structural breaks are present in the second moment structure (see also Jiao et al., 2020). 
{We established the convergence rate of the estimated truncated null distribution, and developed the asymptotic properties of the estimated change point. In addition, we do not assume that the fourth moment of $\{Y_i(t)\colon i\in\mathbb{N}\}$ to be the same across $i$ under the $H_a$, making the theory suitable for a broad range of cases. }

An important motivation and application of our method is structural break detection in brain signals. Comparing with other methods, the proposed 
functional approach has two main advantages. First, it is robust to physiological or machine noise and between-epoch variation since we propose to check the covariance function of complete epoch trajectories. Additionally, the proposed method incorporates intra-curve information, which is potentially informative of structural break in brain signals.
Considering the curse of dimensionality, the methodology requires the sample size to be sufficiently large if we want to detect the structural breaks of the brain signals over a wide frequency band. Appropriate dimension reduction techniques will be considered in the future.


\end{document}